\documentclass[11pt]{article}
\usepackage[a4paper, margin=2.5cm]{geometry}

\usepackage{graphicx} 
\usepackage{hyperref}
\usepackage{booktabs}
\usepackage{placeins}
\usepackage{algorithm}
\usepackage{algpseudocode}
\usepackage{amsmath,amssymb, amsthm}
\usepackage{thm-restate}
\usepackage{amsfonts}
\usepackage{cleveref}
\usepackage{xcolor}
\usepackage[labelfont=bf]{caption}
\usepackage{subcaption}
 \usepackage{parskip}

 \newtheorem{theorem}{Theorem}[section]
 \newtheorem{lemma}[theorem]{Lemma}
  
 \newtheorem{fact}[theorem]{Fact}
 \newtheorem{corollary}[theorem]{Corollary}

\newcommand{\R}{\mathbb{R}}

\newcommand{\eps}{\varepsilon}

\newcommand{\cost}{\textsc{cost}}
\newcommand{\dist}{\text{dist}}
\newcommand{\opt}{\textsc{opt}}

\DeclareMathOperator*{\polylog}{polylog}
\DeclareMathOperator*{\poly}{poly}
\DeclareMathOperator*{\val}{Value_{MP}}
\DeclareMathOperator*{\isolated}{isolated}
\DeclareMathOperator*{\argmax}{arg\,max}

\newcommand{\lpar}{\left(}
\newcommand{\rpar}{\right)}
\newcommand{\lbra}{\left\{}
\newcommand{\rbra}{\right\}}

\DeclareMathOperator*{\nval}{Value}

\DeclareMathOperator*{\countVal}{Count}

\usepackage[colorinlistoftodos,prependcaption,textsize=tiny]{todonotes}
 \newcommand{\david}[1] {}
 \newcommand{\maxD}[1] {}

\title{Faster and Simpler Greedy Algorithm for $k$-Median and $k$-Means}
\author{Max Dupr\'e la Tour, David Saulpic}
\date{}

\begin{document}
 \pagenumbering{gobble} 
\maketitle

\begin{abstract}
Clustering problems such as $k$-means and $k$-median are staples of unsupervised learning, and many algorithmic techniques have been developed to tackle their numerous aspects. 

In this paper, we focus on the class of greedy approximation algorithm, that attracted less attention than local-search or primal-dual counterparts. In particular, we study the recursive greedy algorithm developed by Mettu and Plaxton [SIAM J. Comp 2003]. We provide a simplification of the algorithm, allowing for faster implementation: our algorithm matches the state-of-the-art running time for computing a constant-factor approximation in Euclidean space and graph metrics, and, in addition, is the first near-linear-time to compute a polylogarithmic approximation in Euclidean space.
\end{abstract}

\newpage 
\pagenumbering{arabic}  
\section{Introduction}
Clustering problems such as $k$-median and $k$-means lie at the intersection of applied and theoretical algorithms. Applied, because they are staples of unsupervised learning and are widely used for both classification and data analysis; and theoretical, because they have simple and elegant formulations, serving as testbeds for many algorithmic techniques.

In this paper, we focus on approximation algorithms. In general, we can distinguish three big families of techniques: those based on linear programming (e.g., primal-dual), local search, and greedy. For clustering, the primal dual method is the basis for the most accurate approximation algorithm, where a long line of work reached a $2+\eps$ approximation for $k$-median \cite{2approx} and  5.93 for $k$-means \cite{CharikarCGGLW25}. In addition, primal-dual based algorithms are quite flexible with respect to changes in the objective function — for example, they can be extended to handle the ordered $k$-median problem~\cite{ByrkaSS18}, clustering with outliers~\cite{KrishnaswamyLS18}, or give improved approximation for high-dimensional Euclidean space \cite{EuclideanKM}.
However, those have quite slow running time and are not expected to be applied. Local search is competitive for specific input: it yields an almost linear time approximation scheme in low-dimensional Euclidean space \cite{Cohen-Addad18}, or recover the optimal clustering under some stability assumption \cite{Cohen-AddadS17}. Local search also provides great approximation, as small as $2.836 + \eps$ for $k$-median \cite{Cohen-AddadGHOS22}. One drawback is, here as well, the running time, as evaluating the impact of a single local step takes in general linear time. We note here that, following the initial release of our work, Jiang, Jin, Lou and Lu~\cite{jiang2025localsearchclusteringalmostlinear} presented an implementation of local-search based on Locality Sensitive Hashing, to get an $O(c)$-approximation in time $n^{1+1/c+o(1)}$ in Euclidean space and graph metrics. We discuss this further below.

The third category, greedy algorithm, has been less investigated. The usual advantage of greedy algorithm  is their simplicity, that often highlights structural properties of the problem at hand and allow for fast implementation.
However, for clustering   the two most ``natural" greedy have strong lower bound. Iteratively adding the center that decreases most the cost does not give better than $\Omega(n)$ approximation, and its reversed version, which starts from all points being centers and removing the point that increases least the cost, is a $\Omega(\log n / \log \log n)$-approximation \cite{ChrobakKY06}.  On the other hand, the most famous practical algorithm for solving $k$-means is a randomized greedy, $k$-means++ \cite{ArV07}. Its approximation guarantee is however a super-constant $\Theta(\log k)$, and its running time $O(nkd)$ (in Euclidean space $\R^d$) becomes prohibitive for modern, very large scale data and values of $k$.  
The only greedy known to be a constant-factor approximation is a recursive greedy algorithm, due to Mettu and Plaxton \cite{MPOnlineMedian}. However, this algorithm has a strong reputation of being intricate. In addition, although it is among the first and few constant-factor approximation algorithm, it has not seen subsequent improvement, extension or application. 

In this paper, our goal is to simplify the recursive greedy algorithm of \cite{MPOnlineMedian},  in order to  provide new structural hindsight on clustering problems. We also hope that simplification will allow for more applications. And, indeed, we present several of those, including the first linear time $\polylog k$ approximation algorithm for $k$-means in $\R^d$.

\subsection{Our contribution}

The algorithm works not only for $k$-median and $k$-means, but for the more general $(k,z)$-clustering,  which seeks to minimize the sum of the $z$-th power of the distance from each client to its center ($k$-means is $z=2$, $k$-median $z=1$).

Our main result is a theorem which informally says that the algorithm of \cite{MPOnlineMedian} can be implemented fast, as long as two sub-procedures can be implemented efficiently. We reduce therefore the problem of computing efficiently a solution to $(k, z)$-clustering to those two sub-procedures, which we therefore describe as relaxed as possible -- in order to make the problem easier.
Those two procedures involve computing quantities related to balls in a metric space. One of the relaxation we allow is to consider approximate balls instead: a set $N(x,r)$ is a $c$-approximate ball centered at $x$ with radius $r$ if $B(x,r) \subseteq N(x,r) \subseteq B(x, c \cdot r)$.
We apply this both to Euclidean space, in which case a ball is defined as set of points at distance at most $r$ of $x$, but also to sparse graphs, in which case a ball is the set of edges for which at one endpoint is at distance less than $r$ of a given vertex $x$.

The first procedure's goal is to compute, approximately, the number of points contained in approximate balls centered at input points. The second procedure is a data-structure that allows, given $x$ and $r$, to compute a $c$-approximate ball centered at $x$ with radius $r$; and also to remove that approximate ball from the input set. More formally, our theorem is:

\begin{theorem}[see \Cref{thm:correctness} and \Cref{thm:runningtime}]\label{thm:main}
        Let $(P, \dist)$ be a metric space with aspect-ratio $\Delta$,\footnote{The aspect-ratio is the ratio between the largest distance and the smallest non-zero distance in the metric.} and $c > 5$ be a constant. 
    Suppose there is:
    \begin{itemize}
        \item an algorithm that computes the number of points in approximate balls centered on input point with radius $(2c)^i$  in time $T_{\nval}$, and
        \item a datastructure with preprocessing time $T_{Pre}$ and  that is able to: given $x$ and  $r$, compute a $c$- approximate ball $N(x,r)$ in the current ground set in time $T_N \cdot |N(x, r)|$; and remove an approximate ball $N(x, r)$ from the ground set in time $T_{rm} \cdot |N(x, r)|$.
    \end{itemize}
    
    Then the recursive greedy algorithm can be implemented such that it is a  $\poly(c)$- approximation and has running time $T_{\nval} + O\lpar (T_{Pre} + (T_{rm} + T_N)\cdot |P| ) \cdot \log \Delta\rpar$.
\end{theorem}

Furthermore, this algorithm not only computes a solution to $(k,z)$-clustering, but it also provides an ordering of the input points $p_1, \ldots, p_n$ such that for any $k$, the set $\{p_1, \ldots, p_k\}$ forms an $O(1)$-approximation to $(k,z)$-clustering. This variant of $(k,z)$-clustering is referred to as online \cite{MPOnlineMedian} or incremental \cite{SHENMAIER2016312, CHROBAK2011594}  in the literature. This is a shared feature with algorithms based on $k$-means++ seeding, where each prefix is a $O(\log k)$ approximation.

We apply this theorem in Euclidean space and sparse graph with different datastructure, and get the following corollaries:
\begin{corollary}\label{cor:applications}
    The recursive greedy algorithm can be implemented such that it computes, for any constant $c \geq 5$
    \begin{itemize}
        \item in graphs, a $\poly(c)$-approximation to incremental $(k,z)$-clustering in time $\tilde O\lpar m \log^2(n) \log \Delta\rpar$,
        \item in Euclidean space, a $\poly(c)$-approximation to incremental $(k,z)$-clustering in time $O\lpar n^{1+1/c + o(1)} d \log \Delta\rpar$,    \setlength{\emergencystretch}{1em}
        \item in Euclidean space, a $\polylog(k)$-approximation to  $(k,z)$-clustering in time $\tilde O(nd + n\log \log \Delta)$.
    \end{itemize}
\end{corollary}

We note that, in Euclidean space, the $(n-1, z)$-clustering problem is equivalent to the closest pair problem : as remarked in \cite{bhattacharya2025}, the current tradeoff is $\poly(c)$-approximation in time $n^{1+1/c}$, \cite{AndoniI07, AndoniR15} which is believed to be tight \cite{Rubinstein18}. In this case, our Euclidean constant-factor approximation would be tight too.

We also note that, for the standard $(k,z)$-clustering, it is easy to reduce the aspect-ratio $\Delta$ to $\poly(n)$: we show in appendix how to do so in near-linear time in graphs, and \cite{reduceDiam} show how to do it in time $O(nd + n \log \log \Delta)$ in Euclidean space. Hence, for the standard $k$-means and $k$-median in sparse graphs, the $\Delta$ in the bounds above can be replaced with $n$ -- and, for Euclidean space, this  combined with dimension reduction to replace $d$ with $O(\log k)$ \cite{MakarychevMR19}. The last result of the corollary is phrased to reflect those improvements.

\textbf{Comparison with prior work.}
We see our main contribution as a deeper understanding of the greedy algorithm from Mettu and Plaxton. This results in some novel tradeoff between approximation and running time for $(k,z)$-clustering, that we compare below with other existing techniques. 

We note first that our near-linear algorithm in graphs does not beat the state-of-the-art, in terms of approximation and running time: indeed, an algorithm from \cite{Thorup04} (which uses the primal-dual techniques as subroutine) runs in near-linear time and computes a constant-factor approximation to $k$-median. While the constant is quite small compared to ours ($12$), Thorup himself write that this algorithm "is rather complicated and hence unlikely to be of direct practical relevance". Our contribution here is to present a different algorithm, greedy and much simpler than the one of \cite{Thorup04}, which additionally works for the incremental version of the $(k,z)$-clustering as well.
On top of this, Thorup's result is proven only for $k$-median. While it is likely that the analysis would follow for $k$-means and $(k,z)$-clustering, using recent results on the primal-dual method for $k$-means, our result gives the first formal proof for a fast constant-factor approximation to $k$-means in graphs.

As any set of points in Euclidean space can be turned into a sparse graph in time $O\lpar n^{1+1/c + o(1)} d\rpar$ while preserving the pairwise distance up to a factor $O(c)$ (using spanners \cite{Har-PeledIS13}), a constant-factor approximation for Euclidean inputs running in time $O\lpar n^{1+1/c + o(1)} d\rpar$ follows from Thorup's result. Again, this formally holds for $k$-median, but not $k$-means.
Our contribution for Euclidean space is therefore similar to that for graphs: we present a simpler and natural greedy algorithm, that provably works for the general $(k,z)$-clustering problem.

In the realm of linear time, we are not aware of any sub-polynomial approximation algorithm for Euclidean $(k,z)$-clustering: only the algorithm from \cite{charikar2023simple} runs in near-linear time but yields a $O(k^4)$-approximation. 
The paper \cite{reduceDiam} claims that an algorithm from \cite{Cohen-AddadLNSS20} would produce a $O(\polylog n)$ approximation in near-linear time: however, their proof is not correct. More precisely, 
\cite{Cohen-AddadLNSS20} present a near linear-time algorithm with no approximation guarantee: to achieve a $\polylog(n)$-approximation ratio, they need a rejection-sampling step that adds a $n^{1 + o(1)}$ running-time -- as opposed to our near linear $n \polylog(n)$. The other algorithm indeed runs in near-linear time, but does not have any approximation guarantee: indeed, their multi-tree embedding is shown to preserve squared distances, in expectation, up to a polylog factor. However, this embedding does not yield a metric space (in particular, it does not respect triangle inequality\footnote{The multi-tree embedding consists of the following: build several quadtrees, and define the distance between two points to be the minimum distance accross all quadtrees. This minimum may not respect triangle inequality. As an illustration, consider three points on the line, $a$ at $x=0$, $b$ at $x= 2$ and $c$ at $x=4$. Consider two quadtrees, one with a split at $x=1$ (making all points with $x<1$ very far from all points with $x>1$, say distance $10$, while points with $x < 1$ or $x > 1$ are at their exact distance), the other one with a split at $x=3$. In both trees, $a$ and $c$ are at distance $10$, hence they are at distance 10 in the multi-tree. However, in the first one $b$ and $c$ are at distance $2$, and in the second $a$ and $b$ are at distance $2$: hence, with triangle inequality $a$ and $c$ should be at distance $4$.}): we are not aware of any approximation algorithm in this setting -- in particular, the $k$-means++ algorithm employed by \cite{Cohen-AddadLNSS20} does not.

Our near linear-time approximation algorithm is therefore new, and, importantly, corrects the claim made in \cite{reduceDiam}.

\textbf{Comparison with subsequent work}
Jiang, Jin, Lou and Lu~\cite{jiang2025localsearchclusteringalmostlinear} took inspiration from an initial pre-print of our paper to show how to implement fast local search, in spaces that admit sparse spanners: they get a $O(c)$-approximation for $(k,z)$-clustering in time $n^{1+1/c+o(1)}$. Hence, they improve our tradeoff between run-time and approximation ratio, and achieve the same result as combining spanners and Thorup's algorithm. Compared to the latter, their local search procedure has the advantage being much simpler, and their proof directly generalizes to $k$-means.
Our initial pre-print only mentioned results in Euclidean space, while \cite{jiang2025localsearchclusteringalmostlinear} applies to any metric with sparse spanners. We concurrently generalized our techniques to graphs -- hence, to any metric that admits efficiently computable sparse spanners.

We note that this paper explicitly cites our work as a motivation: we present a $\poly(c)$-approximation in time $n^{1+1/c+o(1)}$, and their goal was to bring the approximation-ratio to $O(c)$, closer to the "LSH tradeoff". 
In addition, we note that their algorithm fails to get any near-linear time run-time, as they deeply rely on LSH.

Another preprint appeared on arxiv \cite{arxiv} after the initial release of the present article. They claim a running time $\tilde O(nd)$: however, in their proof, they have a total running of $\Omega(k^4)$; hence, the claim of near-linear running time is only valid under the assumption $n \geq k^4$ (see their proof of Theorem 3.23).

\subsection{Further related work}

\textbf{Specific to Euclidean space.} The $k$-median and $k$-means problems are NP-hard even when the input is in the Euclidean plane $\mathbb{R}^2$ \cite{megiddo1984complexity, MahajanNV12}. However, in low-dimensional spaces, it is possible to compute a $(1+\eps)$-approximation, for any $\eps > 0$, in time $f(\eps, d) \tilde{O}(n)$ \cite{CohenAddadFS19}. If the target running-time is polynomial in the dimension $d$, the problems becomes NP-hard to approximate: within a factor of 1.015 for $k$-median and 1.06 for $k$-means \cite{Cohen-AddadSL22}. The best approximation ratios are $1+\sqrt{2}$ for $k$-median and 5.912 for $k$-means, based on a primal-dual algorithm running in large polynomial time \cite{Cohen-AddadEMN22}. 
For faster and practical algorithms, \cite{Cohen-AddadLNSS20} improves the running time of $k$-means++ to almost linear, while roughly preserving the approximation guarantee, and \cite{LattanziS19} that improves the approximation guarantee to $O(1)$ albeit with a running time of $O(nkd)$. For linear-time algorithm, the one from \cite{charikar2023simple} stays the best we know of.
Embedding the input metric into a quadtree yields an expected distortion on distances of $\poly(d) \log \Delta$: hence, combined with dimension-reduction to turn $d$ into $\log k$, with reduction of the diameter to $\poly(n)$, one can easily compute a $\polylog(k) \log n$-approximation to $k$-median. However, this does not work for $k$-means as the quadtree does not preserve squared distances (see e.g. \cite{CohenAddadFS19}).

Several sketching techniques are  applicable to clustering in Euclidean space: it is possible to reduce the dimension to $O(\eps^{-2} \log k)$ in near-linear time $\tilde{O}(nd)$, while preserving the cost of any clustering up to a multiplicative $(1 \pm \eps)$ factor. It is also possible to build coresets in time $\tilde{O}(ndk^{o(1)} + n \log \log \Delta)$, which reduces the number of distinct points to $O(k \eps^{-2-z})$ (see \cite{reduceDiam} for the specific running time,\footnote{As we noted previously, this paper uses polylogarithmic approximation in near-linear time, which was actually not known; we state their run-time when using an almost-linear time algorithm instead.} which uses the coreset algorithms from \cite{stoc21, Cohen-AddadLSS22}). Combining those techniques with e.g. \cite{LattanziS19}, it is possible to compute an $O(1)$-approximation to $(k,z)$-clustering in time $\tilde{O}(ndk^{o(1)} + k^2)$. 

\bigskip
\textbf{General metric spaces.}
Beyond Euclidean space, $k$-median is NP-hard to approximate within a factor of $1+2/e$ and $k$-means within $1+8/e$ \cite{GuK99}. 

For the incremental version of $k$-median, the best known approximation ratio is 7.656 for general metric spaces \cite{CHROBAK2011594} and 7.076 for Euclidean spaces \cite{SHENMAIER2016312}. The approximation ratio cannot be better than 2.01 \cite{CHROBAK2011594}.

Last, Mettu and Plaxton  \cite{MPOnlineMedian} introduced an algorithm related to recursive greedy, for the Facility Location problem. In this algorithm as well they weight balls with a value: \cite{FacilityLocationSubLinear} showed how to use a value proportional to the number of points in the ball, to estimate the optimal cost of Facility Location problem in a streaming setting.

\section{Preliminaries}

The $(k,z)$-clustering problem is defined as follows: the input is a metric space $(P,\dist)$ with $|P| = n$, an integer $k$, and a $z \geq 1$. The goal is to find a set of $k$ points $S\subseteq P$ that minimizes $\cost(P, S) := \sum_{x \in P} \dist(x, S)^z,$ where $\dist(x, S) := \min_{s \in S} \dist(x, s)$. We say that a set of $k$ points $C_k$ is an $\alpha$-approximation to $(k,z)$-clustering when $\cost(P, C_k) \leq \alpha \cdot \min_{S, |S|=k} \cost(P, S)$.

A list of $n$ points $c_1, ..., c_n$ is an $\alpha$-approximation to the incremental $(k,z)$-clustering problem on input $P$ when for any $k = 1, ..., n$, the prefix $c_1, ..., c_k$ is an $\alpha$-approximation to $(k,z)$-clustering on $P$.

We assume without loss of generality that the smallest pairwise distance between points of $P$ is $1$, and we let $\Delta$ be an upper bound on the diameter of the input $P$ (i.e., the largest pairwise distance).

We will consider different metric spaces. The first type is metric spaces induced by a positively weighted connected graph. In this setting, $P$ is a subset of the vertices of a graph weighted $G = (V, E, w)$, where $w \colon E \to \mathbb{N} \setminus \{0\}$, $|V| = n$, and $|E| = m$. The distance $\dist(u, v)$ between two vertices $u,v \in V$ is defined as the length of the shortest (weighted) path in $G$.
The second metric is when $P$ is a multiset of points in $\mathbb{R}^d$. In this case, $\dist$ is the usual Euclidean distance.

In our algorithm, we will often use a subroutine to compute the size of a union of sets.

\begin{lemma}\label{lem:sizeUnion}
    Given a set of items $P$, a collections $S_1, ..., S_m$ of subset of $P$, and a collection of queries $Q_1, ..., Q_t \subseteq \{1, ..., m\}$, there is an algorithm with running-time $O\lpar (\sum |S_i| + \sum |Q_i|)\cdot \log t\rpar$ that is able to compute, with probability $1-1/t^2$, an estimate for any $i$ of $\left|\cup_{j \in Q_i} S_j\right|$ correct up to a factor $3$.
\end{lemma}
\begin{proof}
We rely on the sketching technique introduced by \cite{flajolet1985probabilistic, alon1996space}. They show that there is a function $r : \mathbb{R}^d \rightarrow \mathbb{R}$ such that, for any fixed set $U$, $|U|$ is well approximated by $2^{Y_U} := 2^{\max_{u \in U} r(u)}$. Formally, with probability $2/3$, it holds that $\frac{1}{3} \leq \frac{|U|}{2^{Y_U}} \leq 3$. (See Proposition 2.3 in \cite{alon1996space}). The running time to compute the function $r$ is the time to evaluate a pairwise independent hash function, e.g. $O(1)$.

Our algorithm therefore computes, for each $S_j$, the value $S_j := \max_{p \in S_j} r(p)$ in times $O(\sum |S_j|)$. For any $Q_i$, it  holds that $Y_{Q_i} := \max_{j \in Q_i} Y_j$ satisfies with probability $2/3$ that
\begin{align*}
\frac{1}{3} \leq \frac{\left|\cup_{j \in Q_i} S_j\right| }{2^{Y_{Q_i}}} 
\leq 3
\end{align*}

Computing each $Y_{Q_i}$ takes time $O(\sum |Q_i|)$. Therefore, it only remains to boost the probability to ensure the guarantee holds for all $Q_i$ simultaneously: for this, we run  $3\log(t)$ many copies of the algorithm and let $\countVal(Q_i)$ be the median of those estimates. A standard argument shows that, with probability $1 - 1/t^2$, it holds for all $Q_i$ that $\frac{1}{3} \leq \frac{\left|\cup_{j \in Q_i} S_j\right|}{\countVal(Q_i)} \leq 3$, which implies the lemma.

The overall running time is $O\lpar (\sum |S_j| + \sum |Q_i|) \log(t)\rpar$.
\end{proof}

\section{The Greedy Algorithm}

\subsection{Description of the original algorithm}
In what follows, we assume that $z \geq 1$ is fixed. We start by presenting the  original algorithm and definitions from  Mettu and Plaxton~\cite{MPOnlineMedian}. 
\begin{itemize}
    \item Given a ball $B = B(x,r):= \{y\in P, \dist(x,y) \leq r\}$, the \emph{value} of $B$ is $\val(B):= \sum_{y\in B} (r-\dist(x,y))^z$.
    \item A \emph{child} of a ball $B(x,r)$ is any ball $B(y,r/2)$, where $y\in P$ and $\dist(x,y) \leq 10 r$.
    \item For any point $x\in P$ and a set of centers $C$, let $\isolated(x,C)$ denote the ball $B(x,\dist(x,C)/100)$ if $C$ is not empty; and $B(x,\max_{y\in P}d(x,y))$ if $C = \emptyset$. Intuitively, this corresponds to very large ball centered at $x$ that is far away from any center of $C$.\footnote{In those definitions, we chose the scalar constants $2, 10, 100$ for convenience: the whole analysis can be parameterized more carefully in order to optimize the approximation ratio. We opted for simplicity.}
\end{itemize}

The algorithm is a recursive greedy procedure, that starts with $C = \emptyset$ and repeats $n$ times the following steps:  start with the ball $\isolated(x, C)$ with maximum value over all $x \in P$ (with ties broken arbitrarily), and as long as this ball has more than one child (i.e. as long that there are at least two distinct points of the input $P$ "close" to the ball) replace it with the child with maximum value. Let $x$ be the center of the last chosen ball: add $x$ to $C$, and repeat -- see \Cref{alg:mp1} for a pseudo-code.

\begin{algorithm}
\caption{Recursive Greedy}
\label{alg:mp1}
\begin{algorithmic}[1]
\State{Let $C_0 = \emptyset$}
\For{$i$ from $1$ to $n$}
\State{Let $B$ be a maximum value ball in $\{\isolated(x,C_i)| x\in P\setminus C_i\}$}
\While{$B$ contains more than one point}
\State{Replace $B$ by a maximum value child of $B$.}
\EndWhile
\State{$C_{i} = C_{i-1} \cup \{c_{i}\}$, where $c_{i}$ is the center of $B$.}
\EndFor
\end{algorithmic}
\end{algorithm}

\begin{theorem}[\cite{MPOnlineMedian}]\label{thm:mp1}
    For any fixed $z$ and for all $k$, the cost of $C_k$ is a $O(1)$-approximation of the optimal $(k,z)$-clustering cost. 
\end{theorem}

Mettu and Plaxton prove that this algorithm can be implemented in $ O(n^2) $ time, which is linear with respect to the input size (when the metric space is given as a full matrix of pairwise distances). We modify this algorithm to achieve a fast implementation when the metric is described as a sparse graph or a Euclidean space. 
The general idea is that the recursive greedy algorithm still achieves an $O(1)$-approximation even all quantities are approximated: the value function, child sets, and the sets of isolated balls need not be exactly computed. 

\subsection{Simplification and extension}

The main source of conceptual difficulty of \cref{alg:mp1} is the notion of value, which is not easy to grasp intuitively: the algorithm recurses down to denser region, but the notion of denser it uses is not the most natural. 
Our first contribution is to show that the value can be replaced essentially by the number of points inside the ball, times the radius of the ball. With this definition of value, the interpretation of the algorithm is more straightforward: the algorithm recurses down to the children containing most points.

In order to simplify the application of the algorithm, we show in addition that the number of points can be approximated, and that the children of a ball can be computed approximately as well.  We introduce a parameter $c \geq 5$, which governs the trade-off between run-time and approximation ratio.

\textbf{Simplifying the value function.} Instead of $\val(B(x, r)) = \sum_{y \in B(x, r)} (r - \dist(x, y))^z$, we will use a function $ \nval$ approximating $r^z \cdot |B(x, r)|$ as follows: first, count the number of points inside the ball $B(x, r)$ up to a factor $3$, and then multiply by $r^z$. Formally, we show that it is enough to use a function $\nval$ that satisfies
\[\forall x\in P, ~r^z/3 \cdot |B(x, r)| \leq \nval(B(x, r)) \leq 3r^z \cdot |B(x, c\cdot r)|.\]
This is our key new insight, that allows for fast implementation: it is indeed much easier to approximately count the number of points in a ball than to evaluate $\val$.

\textbf{Approximating balls.}
In addition to this key simplification, we allow for some approximations in the computation of the different ball considered by the algorithm.

\textit{Allowing approximation of balls to redefine children.} 
We say that $N(x, r)$ is a $c$-approximate ball of $x$ at radius $r$ if it satisfies $B(x,r) \subseteq N(x,r) \subseteq B(x, c\cdot r)$. When $c$ is fixed, we simply say that $N(x,r)$ is an approximate ball. 
The algorithm uses this notion instead of the "child" used in \Cref{alg:mp1}: it considers all balls (of radius $r/2c$) that are centered at points that are in an approximate ball of $x$ (of radius $r$). 

\textit{Forbidding balls.} 
To select the starting ball at the beginning of an iteration, we move away from the notion of $\isolated$. 
Instead, our algorithm maintain a set of available balls. Initially, those are all balls; when center $c_j$ is placed at the end of $j$-th iteration, our algorithm \textit{forbids} (i.e., remove them from the set of available balls) all balls that are too close to $c_j$. More precisely, the algorithm computes, for all $r$ powers of $2c$ such that $1 \leq r \leq \Delta$, an approximate ball $N(c_j, 100c^4 \cdot r)$. It removes from the set of available balls all balls of the form $B(p, r)$ with $p \in N(c_j, 100c^4 \cdot r)$. There are $O(n \log \Delta)$ such balls.

\textit{Reducing the number of balls.}
Perhaps not surprisingly, the radius of balls can be rounded, in order to limit the number of distinct balls to consider.
For this, we assume for simplicity that the diameter $\Delta$ is a power of $2c$. The algorithm will consider only balls of the form $B(x,r)$, where $x \in P$ is an input point and $r$ is a power of $2c$, such that\footnote{Although points are at distance at least $1$, it is important for our algorithm to consider balls with smaller radius $1/(2c)^7$ in order to be sure that, around any point, there is a ball available unless there is a center at the point.} $1/(2c)^7 \leq r \leq \Delta$. 

We give the pseudocode of our modified algorithm in \cref{alg:main}. 

\begin{algorithm}\caption{Modified recursive greedy} \label{alg:main}
\textbf{Input:} A metric space $(P,\dist)$ and a number of clusters $k$.\\
\textbf{Output:} a set of $C$ of at most $k$ centers.
\begin{algorithmic}[1]
\State{Define the set of \textit{available} balls to be $\lbra B(x, \Delta/(2c)^\ell), x \in P, l\in \{0, \dots,\log_{2c}(\Delta)+7\} \rbra$.}
\State{Compute $\nval(B(x,r))$ for all the available balls.}
\For{$i$ from $1$ to $k$}
\If{The set of available ball is empty\footnotemark}
\State{return the solution $C_{i-1} = \{c_1,\dots,c_{i-1}\}$.}
\EndIf 
\State{Let $B = B(x,r) $ be an available ball with largest $\nval$}
\While{$r > 1/(2c)^7$}
\Comment{\textit{Center Selection loop}}
\State{Compute an approximate ball $N = N(x, 10c \cdot r)$}
\State{Update $x$ and $r$: select $x \in \argmax_{y \in N} \nval\lpar B\lpar y, \frac{r}{2c}\rpar\rpar$ and $r \gets \frac{r}{2c}$}
\EndWhile
\State{$c_i \gets x$}
\For{all radius $r \in \lbra \Delta/(2c)^\ell, \ell \in \{0,..., \log_{2c}(\Delta)+ 7\}\rbra$}
\Comment{\textit{Forbidding loop}}
\For{$x \in N\lpar c_i,100c^4\cdot r\rpar$}
\State{Remove $B(x,r)$ from the set of available balls.}
\EndFor
\EndFor
\EndFor
\State{Output the solution $C_k = \{c_1,\dots,c_k\}$. }
\end{algorithmic}
\end{algorithm}
\footnotetext{In that case, as explained in the footnote 2, $i$ is more than the number of distinct points.}

\subsection{Analysis}
To clarify the analysis, we stop the algorithm after $k$ iterations and prove that the set of centers $C_k$ is a $\poly(c)$-approximation of the optimal $(k,z)$-clustering. However, the algorithm does not depend on $k$, and therefore the set of centers after $k'$ iterations for $k' \leq k$ is also a $\poly(c)$-approximation of the optimal $(k',z)$-clustering. In particular, if we modify the algorithm to stop after $n$ iterations instead, it provides a $\poly(c)$-approximation of the \emph{incremental} $(k,z)$-clustering problem. 

For clarity, we split the proof of \Cref{thm:main} into two parts: first the approximation guarantee, then the running time.

\begin{theorem}\label{thm:correctness}
    For any $k$, the set of centers output by the simplified recursive greedy \Cref{alg:main} gives a $\poly(c)$-approximation of the optimal $(k,z)$-clustering solution.
\end{theorem}

We provide the full proof of this theorem in the appendix~\ref{appendix:correctness}, and focus here only on the running time. 

In our applications, computing the values of ball turns out to be easy -- since it boils down to estimating the number of points in balls, which is a commonly studied task. Hence, the main remaining task when applying this algorithm to a specific metric space is to bound the running time of the Forbidding loop and of the Center Selection loop.

\begin{theorem}\label{thm:runningtime}
    Let $(P, \dist)$ be a metric space. 
    Suppose there is an algorithm that computes all values in time $T_{\nval}$ and a datastructure with preprocessing time $T_{Pre}$ and  that is able to, for any fixed $r$: for any given $x$, remove all points of a $c$-approximate ball $N(x,r)$ from the ground set in amortized time $T_{rm} \cdot |N(x,r)|$, and compute a $c$-approximate ball $N(x, r)$ of the current ground set in time $T_N \cdot |N(x,r)|$.
    
    Then the running time of \Cref{alg:main} is bounded by 
    \begin{equation*}
        T_{\nval} + O\lpar \log \Delta ( T_{Pre} + |P|(T_N + T_{rm}) )  \rpar
    \end{equation*}
\end{theorem}

To prove this theorem, it suffices to bound the complexity of the \textbf{Center Selection} loop and the \textbf{Forbidding} loop, since by definition the value of all balls can be computed in time $T_{\nval}$.

In the \textbf{Center Selection} loop, we will show that each ball is considered at most once when the algorithm selects the ball with the maximum value at line 10.

In the \textbf{Forbidding} loop, we use a data structure that maintains $O(\log n)$ copies of the input--one for each radius considered by the algorithm. After a center is selected, for each radius $r$, we forbid all balls $B(x, r)$ with $x \in N(c_i, 100c^4 \cdot r)$. These points are deleted from the corresponding data structure to prevent them from being reconsidered later in the Forbidding loop when subsequent centers are selected.

\subsection{Running time analysis of the Center Selection loop}

We now introduce a new definition. We say that a ball $B'$ is a \emph{potential descendant} of a ball $B$ if there exists a sequence of balls $B_0, \dots, B_{\ell}$, with $B_i = B(x_i, r_i)$, such that $B_0 = B$, $B_{\ell} = B'$, and for all $i$, $\dist(x_i, x_{i+1}) \leq 10c^2 r_i$ and $r_{i+1} = r_i / 2c$. 

Note that if two balls $B$ and $B'$ appear in the same center selection loop with $B'$ appearing after $B$, then $B'$ is a potential descendant of $B$.

\begin{fact}
\label{lem:distDesc}
    If $B(y, r_y)$ is a potential descendant of $B(x, r_x)$, then $\dist(x,y) \leq 20c^2 \cdot r_x$.
\end{fact}
\begin{proof}
    Let $B(y, r_y)$ be a descendant of $B(x, r_x)$, and let $x_0 = x, ..., x_\ell = y$ be the centers of the sequence of balls $(B_0, \dots, B_{\ell})$ from the definition of descendant. 
    For every $i$, we have $\dist(x_i, x_{i+1}) \leq \frac{10c^2\cdot r_x}{(2c)^i}$. By triangle inequality, this implies 
    \begin{align*}
        \dist(x,y) &\leq \sum_{i = 0}^{\ell-1} \frac{10c^2\cdot r_x}{(2c)^i} \leq 20 c^2 \cdot r_x.
    \end{align*}
\end{proof}

A consequence of this fact is that any ball appears at most once in an approximate ball of the Center Selection loop:
\begin{lemma}\label{lem:onlyOnce}
    For any $x \in P, r \in \R^+$, the ball $B(x, r)$ appears at most once in the center selection loop.
\end{lemma}
\begin{proof}
We split the proof into two parts: first, if a ball appears in an approximate ball, then it is forbidden at the next Forbidding step. Second, if a ball is available, then all its potential descendants are available. Combined, those two results conclude our lemma.

\begin{fact}\label{cor:childForbidden}
    If a ball $B(x, r_x)$ appears line $9$ as a candidate child in the approximate ball $N(y,10c \cdot r_y)$ computed from a ball $B(y,r_y)$ during the $j$-th iteration of the center selection loop, then $B(x, r_x)$ is forbidden during the forbidding procedure after $c_j$ is selected.
\end{fact}
\begin{proof}
The radius of $B(y,r_y)$ is $r_y = 2c \cdot r_x$. \Cref{lem:distDesc} ensures that  $c_j$ is at distance at most $20c^2 \cdot 2c \cdot r_x$ from $y$. Moreover $\dist(y,x) \leq 10c^2 \cdot 2cr_x = 20c^3r_x$. Therefore, $\dist(x, c_j) \leq 60c^3 r_x < 100 c^4 r_x$, hence the ball $B(x, r_x)$ is in the set $N(c_j, 100c^4 r_x)$ and is forbidden on line 12 when $c_j$ is selected as a center.
\end{proof}

\begin{fact}\label{lem:childAvailable}
    If, at the beginning of an iteration of the loop line 4, a ball $B(x, r_x)$ is available, then all its potential descendants are available.
\end{fact}
\begin{proof}
   Let $B(x, r_x)$ be any ball, and let $B(y, r_y)$ be a potential descendant of $B(x, r_x)$. 
   Suppose that $B(y, r_y)$ is not available at the beginning of an iteration of the loop in line 4. Then there exists a center $c_i$ selected by the algorithm such that $y \in N(c_i, 100c^4 \cdot r_y)$ and therefore $\dist(y, c_i) \leq 100c^5 \cdot r_y$. Because $B(y, r_y)$ is a descendant of $B(x, r_x)$, we know by \Cref{lem:distDesc} that $\dist(x, y) \leq 20c^2 \cdot r_x$. Using the triangle inequality, we get $\dist(x, c_i) \leq \dist(x, y) + \dist(y, c_i) \leq 20c^2 \cdot r_x + 100c^5 \cdot r_y$. We also know that $r_y \leq r_x / 2c$ and therefore $\dist(x, c_i) \leq (50c^4 + 20c^2) \cdot r_x \leq 100c^4 \cdot r_x$. Hence $x \in B(c_i, 100c^4 \cdot r_x) \subset N(c_i, 100c^4 \cdot r_x)$, and $B(x, r_x)$ is also forbidden at the $i$-th iteration.
\end{proof}
Combining those two facts concludes the proof.
\end{proof}

\begin{proof}[Proof of \Cref{thm:runningtime}]
    First, by definition, the running time of computing all values in line 2 of the algorithm takes time $T_{\nval}$.  

    For analyzing the Forbidding loop, we note the following: instead of computing $N(c_i, 100c^4 r)$ in line 14, it is merely enough to compute the set of balls that have not been forbidden yet in that set. To do so with the datastructure from the theorem statement, one can merely do the following: for each valid $r$, initialize the datastructure with all points in the ground set. 
    Then, after each computation of an approximate ball $N(c_i,100c^4 r)$ of the current set, remove all $x \in N(c_i,100c^4 r)$ from the ground set of the corresponding $r$ in line 15.
    Hence, the preprocessing time is $T_{Pre} \cdot O(\log \Delta)$, and the total running time is $O((T_N + T_{rm})|P| \log \Delta)$. Indeed, there are $O(\log \Delta)$ different radius $r$, and at each level any point appears at most once before being removed. The guarantees of the datastructure therefore ensure that, at any given level, the total running time is bounded by $O((T_N + T_{rm})|P|)$.

    Last, it remains to analyze the Center Selection loop. For this, we can use the datastructure from the lemma, without removing any point: \Cref{lem:onlyOnce} ensures that each ball $B(x, r)$ appears only once. Hence, the total running time is again $O\lpar (T_{Pre} +  T_N |P|) \log \Delta) \rpar$.
\end{proof}

\section{Implementation in the Euclidean setting}

\subsection{Near-linear time approximation via multiple quadtrees}

Tree embeddings are a common tool for designing approximation algorithms on metric spaces. A particularly useful tree embedding in Euclidean spaces is the \emph{quadtree}; see, for example, \cite{har2011geometric, Arora98} for general reference, and \cite{KolliopoulosR07, Cohen-AddadLNSS20} for applications to clustering.

In a quadtree, the input space is enclosed in an axis-aligned hypercube that is recursively subdivided into $2^{d}$ smaller hypercubes of half the side length. 

If the entire input is translated by a uniformly random vector in $[0, \Delta]^{d}$, then the standard analysis of distortion induced by the quadtree (see e.g. \cite{reduceDiam}) shows that with probability at least $1/2$, the smallest quadtree cell containing both $p$ and $q$ has side length at most $4 \sqrt{d} \cdot \|p - q\|$. We define $\dist_T(p, q)$ as the diagonal of the smallest quadtree cell containing both $p$ and $q$.

To boost the probability, one can take $t = O(\log n)$ independent random shifts and construct $t$ quadtrees, with the following guarantees:

\begin{itemize}
    \item for all quadtree $T$, and any $p,q$, $\|p-q\| \leq \dist_T(p, q)$
    \item with probability $1-1/n^2$, for any $p, q$ there exist a quadtree $T$ such that $\dist_T(p, q) \leq 4 d \|p-q\|$.
\end{itemize}
The construction of a single quadtree takes time $O(nd \log \Delta)$ \cite{Cohen-AddadLNSS20}, hence the total construction time is $O(nd \log n \log \Delta)$. 

We will use these quadtrees to run Algorithm~\ref{alg:main} with parameter $c = 4d$. In order to evaluate the value of all balls of radius $r$, we compute all quadtree cells of side length $L = 4 \sqrt d \cdot r$. Now, to compute an approximate ball $N(x,r)$, it is merely enough to take the union of all 
cells with side length $L$ that contain $x$. Indeed, with probability $1$ all points in those cells are at distance at most $L\sqrt{d}$ of $x$; and, with probability $1-1/n^2$, any point at distance $r$ is in the same cell as $x$ in one of the quadtrees.

Therefore, computing the set $N(x,r)$ can be done in time $O(\log n) |N(x,r)|$, hence $T_N = O(\log n)$ -- as there are $O(\log n)$ different quadtrees, the union can be computed with this running time. Removing a point from the datastructure simply takes time $T_{rm} = O(\log n)$, to remove it from each of the $O(\log n)$ quadtree.

In addition, computing the size of $N(x,r)$ is a direct application of \Cref{lem:sizeUnion} -- and computing $N(x, r)$ for all point $x \in P$ and $r$ power of $2$ takes time $T_{\nval} = O(n \log \Delta \log^2 n)$. 

Theorem~\ref{thm:runningtime} directly implies:

\begin{corollary}
In Euclidean space, the recursive greedy algorithm can be implemented to provide a $\poly(d)$-approximation to incremental $(k,z)$-clustering in time $O(n d \log n \log \Delta + n \log \Delta \log^2 n)$.
\end{corollary}

Note that, using dimension-reduction, the dimension for $(k,z)$-clustering can be reduced to $O(\log k)$ \cite{MakarychevMR19} in time $O(nd \log d)$. This, with the above corollary, shows the last two items of \cref{cor:applications}.

\subsection{Constant-factor approximation via Locality-sensitive hashing:}
The tool we use in the Euclidean setting is Locality-sensitive hashing \cite{AndoniI07}. The precise result we use is the following:
\begin{lemma}[See section D in \cite{Cohen-AddadLNSS20}]\label{lem:lsh}
Let $P \subseteq \R^d$, $r \in \R^+$, and $\ell = (n/\delta)^{1/c^2}$ ; there is a family of hash functions from $\R^d$ to some universe $U$ such that, with probability $1-\delta$, 
if $f_1, ..., f_\ell$ are drawn from this family:
\begin{itemize}
\item For any $p,q \in P$ with $\dist(p,q) \geq c\cdot r$, then for all $i = 1, ..., \ell$ $f_i(p) \neq f_i(q)$
\item For any $p,q \in P$ with $\dist(p,q) \leq r$, then there exists $i \in \{1, ..., \ell\}$ with $f_i(p) = f_i(q)$.
\end{itemize}
Furthermore, the hash functions satisfy the following:
\begin{itemize}
\item for any $i, p \in \R^d$, computing $f_i(p)$ takes time $O\lpar dn^{o(1)}\rpar$,
\item after preprocessing time $O\lpar \ell d \cdot n^{1+o(1)} \rpar$, one can compute for any $i, p$ the set $T_i[u] := \{p: f_i(p) = u\}$ in time $O(|T_i[u]|)$.
\end{itemize}
\end{lemma}

We use the previous lemma in two ways: first, it allows us to compute an approximate neighborhood of each point quickly, and second, combined with streaming techniques, to estimate the size of this neighborhood efficiently. We start with the former (where we replaced, for simplicity of notation, the success probability $1-\delta$ with $1-1/n^2$):

\begin{corollary}\label{cor:neighborhood}
For any $r \in \R^+$ and $P \subseteq \R^d$, there is a datastructure with preprocessing time $T_{pre} = O\lpar dn^{1+3/c^2+o(1)}\rpar$ that can, with probability $1-1/n^2$:
\begin{itemize}
    \item remove a point from $P$ in time $O\lpar n^{3/c^2}\rpar$
    \item answer the following query: 
for any point $p\in P$, compute a set $N(p, r)$ of points of $P$ such that $B(p, r)\cap P \subseteq N(p, r) \subseteq B(p, c\cdot r) \cap P$. The query time is $O\lpar n^{3/c^2} |N(p,r)|\rpar$.
\end{itemize}
\end{corollary}
\begin{proof}This is a direct application of \Cref{lem:lsh}: given $r$ and $\delta = 1/n^2$, compute $f_i(p)$ for all $i$ and $p$, in time $O\left(dn^{1+3/c^2+o(1)}\right)$. First, to remove a point $p$ from $P$, simply remove it from all the tables $T_i[f_i(p)]$ for $i = 1, ..., \ell$: this takes time $O(\ell) = O\left(n^{3/c^2}\right)$.

To answer a query given a point $p$, compute $T_i[f_i(p)]$ for all $i$, in time $O(|T_i[f_i(p)]|)$ and define $N(p, r) := \cup_{i=1}^\ell T_i[f_i(p)]$. The running time to compute the union is at most $\ell \cdot O(|N(p, r)|) = O(n^{3/c^2} |N(p, r)|)$. The first two bullets of \Cref{lem:lsh} ensure the desired accuracy guarantee. 
\end{proof}

Hence, in the vocabulary of \Cref{thm:runningtime}, $T_{rm} = T_N = O\lpar n^{3/c^2}\rpar$. It only remains to compute the values: this can be easily done combining \Cref{lem:lsh} with the sketching techniques of \Cref{lem:sizeUnion}, as follows:
\begin{lemma}\label{lem:computeVal}
Given a radius $r$, there is an algorithm that runs in time $O\lpar dn^{1+3/c^2 + o(1)}\rpar$ and computes, for all $p \in P$, $\nval(B(p, r))$ such that, with probability $1-1/n^2$, it holds that $\forall p, ~r^z \cdot |B(p, r)\cap P|/3 \leq \nval(B(p, r)) \leq 3r^z \cdot |B(p, c\cdot r)\cap P|$. 
\end{lemma}
\begin{proof}
We show how to compute, for all $p \in P$, an approximate count of the number of points in $B(p, r)$, namely a value $\countVal(p, r)$ such that $ |B(p, r) \cap P|/3 \leq \countVal(p, r) \leq 3r \cdot |B(p, c \cdot r) \cap P|$. Multiplying $\countVal$ by $r^z$ gives the lemma.

To build the estimates $\countVal(p, r)$, the first step of the algorithm is to compute $f_i(p)$, for all $i \in \{1, ..., \ell\}$ and all $p \in P$, using \Cref{lem:lsh} with $r$ and $\delta = 1/n^2$. This takes time $O\lpar dn^{1+3/c^2 + o(1)}\rpar$. 
Due to \Cref{lem:lsh}, we have the guarantee that, with probability $1-1/n^2$, 
\begin{align*}
|B(p, r)\cap P| \leq \left|\cup_{i=1}^\ell T_i[f_i(p)]\right| \leq |B(p, c\cdot r)\cap P|.
\end{align*}
Therefore, it is merely enough to estimate $\left|\cup_{i=1}^\ell T_i[f_i(p)]\right|$ using \Cref{lem:sizeUnion} (with $S_{i, u} = T_i[u]$, $t = n$, and for all $p$, $Q_p = \{f_i(p), i = 1, ..., \ell\}$).
As each point $p$ is in at most $\ell$ sets $S_{i,u}$, and each $Q_p$ has size at most $\ell$, the running time of the algorithm from \Cref{lem:sizeUnion} is $ O(\ell n)$

Hence, the overall running time is $O\lpar dn^{1+3/c^2 + o(1)}\rpar + O(n  \ell \log(n)) = O\lpar dn^{1+3/c^2 + o(1)}\rpar$.
\end{proof}

Thus, Theorem~\ref{thm:runningtime} implies:

\begin{corollary}
In Euclidean space, the recursive greedy algorithm can be implemented to provide a $\poly(c)$-approximation to incremental $(k,z)$-clustering in time $O(n^{1+3/c^2 + o(1)} d \log \Delta)$.
\end{corollary}

\section{Near-linear time implementation for Graphs}\label{sec:graphs}

In sparse graphs, we can use black-box a set of results to compute the values and approximated balls. First, we fix some notation: for a vertex $u$ and a radius $r$, the ball $B_v(u, r)$ is the set of all vertices at distance at most $d$ from $u$ ; and the ball $B_e(u, r)$ is the set of all edges with at least one endpoint in $B_v(u, r)$.

In order to apply \cref{thm:runningtime}, we will bound the running time to compute exactly $B_v(u, r)$, to compute a $(1+1/10)$-approximation of $|B_v(u, r)|$ and to remove efficiently vertices in a $(1+1/10)$-approximate ball. The running time will depend on $|B_e(u, r)|$, as can be naturally expected in graphs: in our theorem statement, $P$ is the graph itself and therefore $|P|$ is its number of edges.

First, to compute a ball $B_v(u, r)$ we simply use Dijkstra's algorithm: it considers only once each edge adjacent to a vertex at distance less than $r$ of $u$, hence the running time is precisely $O(\log n) |B_e(u,r)|$ (note that there is no need for approximate balls at this step and the computation can be done exactly).

To compute the values of all balls, we use the following result of Edith Cohen, that allows to approximate the number of vertices in all balls in near-linear time:

\begin{lemma}[Theorem 5.1 in \cite{COHEN1997441}]\label{lem:cohen}
There exists an algorithm that takes as input a metric induced by an edge-weighted graph $G = (V, E,w)$ with $n$ vertices and $m$ edges. The algorithm has expected preprocessing time $O(m \log^2 n + n \log^3 n)$ and allows queries $\tilde{n}_v(u, r)$ for any pair $(u, r) \in V \times \mathbb{R}^+$ that estimate the number of points in the ball $B_v(u, r)$ such that:
\begin{itemize}
    \item The expected query time is $O(\log \log n)$.
    \item With probability $1 - O(1/\mathrm{poly}(n))$, for all $(u, r) \in V \times \mathbb{R}^+$, 
  \[
  \frac{\big||B_v(u, r)| - \tilde{n}_v(u, r)\big|}{|B_v(u, r)|} \leq 1/10.
  \]
\end{itemize}
\end{lemma}

Hence, using this result, we can directly compute the values of all the $O(n \log \Delta)$ balls, in expected time $T_{\nval} = O\lpar m \log^2 n + n \log^3 n + n \log \Delta \log \log n\rpar$.

Thus, it only remains to remove points from approximate balls. For this, we rely on a recent result of \cite{GorkiewiczK25}, who showed:
\begin{lemma}[Theorem 4.1 in \cite{GorkiewiczK25}]\label{lem:incrementalSP}
    Let $G = (V, E)$ be a digraph with edge weight in $\{0\} \cup [1,W]$, and let $s \in V$. There exists a deterministic datastructure explicitly maintaining distances estimate $d : V \rightarrow \mathbb{R}_{\geq 0}$ satisfying $\forall v, \dist_G(v, s) \leq d(v) \leq (1+1/10)\dist_G(v, s)$, and supporting edge insertion of source edges $e = (s, v)$.

    The total update time of the datastructure is $O(m\log(nW) \log^2 n)$, where $m$ is the final number of edges of $G$.
\end{lemma}
We use this theorem as follows: first, starting from our undirected graph $G$, we duplicate edges to have a copy of each in each direction. Then, we add an extra source vertex $s$, not connected to any other vertex. We instantiate a copy of the datastructure of the theorem. To remove a $11/10$-approximate ball centered at $x$ of radius $r$, we merely add an edge $(s,x)$ with weight 0, and update distances using the datastructure: the points that the algorithm needs to remove (and have not been already removed) are all those whose estimated distance to $s$ decreases below $r$. Since the datastructure maintains explicitly all distances to $s$, it is easy to compute them with the same running time as the update -- hence, a total update time $O(m\log(nW) \log^2 n)$.

Combining Dijkstra's algorithm with those of \cref{lem:cohen} and \cref{lem:incrementalSP} concludes therefore the implementation of our algorithm in graphs, with a total expected running time of $OO(m\log^2(nW) \log^2 n)$.

\bibliography{biblio}

\appendix
\section{Proof of Correctness}\label{appendix:correctness}

The goal of this section is to prove Theorem~\ref{thm:correctness}. We start with a simple property of the function $\nval$.
\begin{lemma}\label{lem:nvaldecreasing}
For any point $x \in P$ the function $\ell \mapsto \nval(B(x,\Delta/(2c)^\ell))$ is decreasing.
\end{lemma}
\begin{proof}
 Let $r,r'$ be two power of $2c$ such that $r \geq 2c\cdot r'$. We have: 
    \begin{align*}
    \nval(B(x, r)) &\geq r^z \cdot \frac{|B(x, r) |}{3} \\
    &= 3 \frac{r^z}{(2c)^z} \cdot |B(x, r)| \cdot \frac{(2c)^z}{9}\\
    &\geq 3 r'^z \cdot |B(x, c \cdot r')| \cdot \frac{(2c)^z}{9} \\
    &\geq \nval(B(x, c \cdot r')) \cdot \frac{(2c)^z}{9} \\
    &\geq \nval(B(x, r')).
\end{align*}
 The last inequality comes from $c \geq 5$ and $z \geq 1$.
\end{proof}

In what follows, we consider a fixed set of $k$ centers $\Gamma \subseteq P$. Our objective is to compare the cost of $C$ output by algorithm~\ref{alg:main} with the cost of $\Gamma$ and demonstrate that $\cost(P, C) \leq O(poly(c)) \cdot \cost(P, \Gamma)$. By setting $\Gamma$ as the optimal $(k,z)$-clustering solution, we can finalize our analysis. It is important to note that $\Gamma$ is restricted to be a subset of the input $P$. In the Euclidean setting, the centers of a solution are typically not required to be part of the input. However, it is well known that the optimal $(k,z)$-clustering, constrained to be a subset of the input, is a $2^z$-approximation of the optimal $(k,z)$-clustering that allows centers to be placed outside of the input points; thus this assumption make us lose a mere factor $O(2^z)$.

For $\gamma \in \Gamma$, let $P_\gamma$ be the cluster of $\gamma$, consisting of all points in $P$ assigned to $\gamma$ in $\Gamma$. We analyze the cost of each cluster independently as follows. We split $\Gamma$ into two parts: $\Gamma_0$ and $\Gamma_1$. $\Gamma_0$ is the set of $\gamma \in \Gamma$ such that no ball in 
 $\left\{ B\left(\gamma, \frac{\Delta}{(2c)^\ell}\right) \mid \ell \in \{0, \dots, \log_{2c}(\Delta) + 7\} \right\} $
is available at the end of the algorithm. Let $\Gamma_1 = \Gamma \setminus \Gamma_0$.

\paragraph{The easy case, dealing with $\Gamma_0$:}
The next lemma shows that if a ball $B\lpar x,\frac{1}{(2c)^7} \rpar$ is not available, then $x \in C$. This directly implies that centers of $\Gamma_0$ are also in $C$.
\begin{lemma}
    Let $x\in P$, if the ball $B\lpar x,\frac{1}{(2c)^7}\rpar $ is not available at the end of the algorithm, then there exists a center $c_j\in C$ such that $c_j = x$.
\end{lemma}
\begin{proof}
We know that there exists a center $c_j \in C$ such that 
\[x \in N\left(c_j, 100c^4 \cdot \frac{1}{(2c)^7}\right) \subseteq B\left(c_j, c \cdot \frac{100c^4}{(2c)^7}\right) = B\left(c_j, \frac{100}{128 c^2}\right).\]
Therefore, $\dist(c_j, x) \leq \frac{100}{128 c^2} < 1$. However, both $c_j$ and $x$ are in $P$, and the minimum distance between two distinct input points is assumed to be $1$. Therefore, $c_j = p$.
\end{proof}

\begin{corollary}\label{lem:gamma0}
    $\Gamma_0 \subseteq C$. In particular for all $\gamma \in \Gamma_0$, we have $\cost(P_\gamma,C) \leq \cost(P_\gamma, \Gamma)$.
\end{corollary}

\begin{corollary}\label{cor:gamma0}
    If none of the balls are available at the end of the algorithm, $\cost(P,C) = 0$. 
\end{corollary}

We can specifically apply \Cref{cor:gamma0} if the algorithm terminates early at line 22 (i.e., before selecting $k$ centers) because none of the balls are available. For the remainder of the proof, we assume the algorithm terminates after selecting $k$ centers and we will therefore denote $C_k = \{c_1, \dots, c_k\}$ the output of the algorithm.

\paragraph{First step to bound the $\cost(P_\gamma,C_k)$ for $\gamma \in \Gamma_1$:}
The main task is to demonstrate that clusters in $\Gamma_1$ are also well approximated. For any center $\gamma \in \Gamma_1$, let $B(\gamma, r_\gamma)$ be the largest ball centered on $\gamma$ that remains available at the end of the algorithm. Such a ball exists by the definition of $\Gamma_1$. We divide the cluster $P_\gamma$ into two parts: $In(P_\gamma) := P_\gamma \cap B(\gamma, r_\gamma)$ and $Out(P_\gamma) := P_\gamma \setminus In(P_\gamma)$.

By the definition of $r_\gamma$, we know that there exists a center in $C_k$ ``not too far" from the ball $B(\gamma, r_\gamma)$ -- as otherwise, a larger ball would be available. This allows us to bound the cost of $Out(P_\gamma)$ in the clustering $C_k$. Furthermore, we can relate the cost of $In(P_\gamma)$ to the value of $B(\gamma, r_\gamma)$, as demonstrated in the following lemma.

 \begin{restatable}{lemma}{inAndOut}\label{lem:inAndout}
For all $\gamma \in \Gamma_1$ , we have:
    \begin{align}
        &\cost(In(P_\gamma),C_k)  \leq  2^{z-1}\lpar(200c^6)^z + 1 \rpar \cdot 3  \nval(B(\gamma,r_\gamma))\label{lem:In}\\
        &\cost(Out(P_\gamma), C_k) \leq 2^{z-1}((200c^6)^z+1) \cdot  \cost(Out(P_\gamma), \Gamma) \label{lem:Out}
    \end{align}
\end{restatable}
\begin{proof}
Fix a $\gamma \in \Gamma_1$.
For any $x\in P_\gamma$, we have 
\begin{equation*}
    \cost(x, C_k) = \dist(x,C_k)^z \leq  (\dist(\gamma,C_k) + \dist(\gamma,x))^z \leq 2^{z-1} (\dist(\gamma,C_k)^z + \dist(\gamma,x)^z)
\end{equation*}
Thus, the first step of the proof is to establish the existence of a center in $C_k$ at a distance of $O(r_\gamma)$ from $\gamma$. By the definition of $r_\gamma$, the ball $B(\gamma, 2c \cdot r_\gamma)$ is not available. Therefore, there is a point $c_j \in C_k$ such that $\gamma \in N(c_j, 200c^5 \cdot r_\gamma)$, and $\dist(\gamma, c_j) \leq 200c^6 \cdot r_\gamma$.

We can now bound the cost of $In(P_\gamma)$ and prove \Cref{lem:In}. If $x \in B(\gamma, r_\gamma)$, we have  $\dist(\gamma,x) \leq  r_\gamma$, and therefore $\cost(x, C_k) \leq 2^{z-1}\lpar(200c^6)^z + 1 \rpar\cdot r_\gamma^z $. Summing this inequality over all $x \in In(P_\gamma)$ yields
\[\cost(In(P_\gamma), C_k) \leq 2^{z-1}\lpar(200c^6)^z + 1 \rpar \cdot \sum_{x \in B(\gamma, r_\gamma)} r_\gamma^z \leq 2^{z-1}\lpar(200c^6)^z + 1 \rpar \cdot  3 \nval(B(p_\gamma, r_\gamma)).\]

We turn to $Out(P_\gamma)$. If $x$ is outside $B(\gamma, r_\gamma)$, we have $\dist(\gamma,x) \geq r_\gamma$. Hence
\begin{align*}
     \cost(x, C_k) &\leq 2^{z-1} ((200c^6 \cdot r_\gamma)^z + \dist(\gamma,x)^z) \\
     &\leq 2^{z-1} ((200c^6 \dist(\gamma,x))^z + \dist(\gamma,x)^z) \\
     &\leq 2^{z-1}((200c^6)^z+1)\cdot \dist(\gamma,x)^z\\
     &\leq 2^{z-1}((200c^6)^z+1)\cdot \cost(x,\Gamma)
\end{align*}
Summing this inequality over all $x\in Out(P_\gamma)$ we get \Cref{lem:Out}
\end{proof}

\Cref{lem:inAndout} shows that points in $Out(P_\gamma)$ have roughly the same cost in the solution $\Gamma$ as in $C_k$, up to a factor of $O(poly(c))$, and the cost of points in $In(P_\gamma)$ is bounded by $O(poly(c)) \cdot \sum_{\gamma \in \Gamma_1} \nval(B_\gamma)$. Therefore, we only need to bound this sum of values.

\subsection{Bounding the sum of values}

To do so, we start by showing a simple lemma that provides a lower bound for the cost of the balls that do not intersect $\Gamma$. We say that a ball $B(x,r)$ is \emph{covered} by $\Gamma$ if $B(x,2c\cdot r) \cap \Gamma \neq \emptyset$.

\begin{lemma}\label{lem:cover}
    If a ball $B(x,r)$ is not covered by $\Gamma$, then $\cost(B(x,c\cdot r) , \Gamma) \geq c^z/3 \cdot \nval(B(x,r))$.
\end{lemma}

\begin{proof}
    Consider $B(x, r)$, a ball not covered by $\Gamma$. Here, $\dist(x, \Gamma) \ge 2c\cdot r$. For any $p \in B(x,c\cdot r)$, the triangle inequality implies $\dist(p, \Gamma) \ge \dist(x,\Gamma) - \dist(x,p) \geq 2c\cdot r - c\cdot r = c\cdot r$. Raising both sides to the power of $z$ and summing for all $p\in B(x,c\cdot r)$, we get $\cost(B(x,c\cdot r), \Gamma) = \sum_{p\in B(x,c\cdot r)} \dist(p,\Gamma)^z \geq \sum_{p\in  B(x,c\cdot r)} (c\cdot r)^z \geq c^z \cdot r^z \cdot |B(x,c\cdot r)| \geq  c^z/3 \cdot \nval(B(x,r))$.
\end{proof}

\paragraph{From $In(P_\gamma)$ to uncovered balls:}

The strategy for bounding the sum of values $\sum_{\gamma \in \Gamma_1} \nval(B(\gamma, r_\gamma))$ relies on the preceding lemma. Our objective is to match each $B(\gamma, r_\gamma)$ (for $\gamma \in \Gamma_1$) with a ball $B(\phi(\gamma), r_{\phi(\gamma)})$ that satisfies the following properties: the balls $B(\phi(\gamma), r_{\phi(\gamma)})$
\begin{enumerate}
    \item are uncovered,
    \item have at least the same value as the balls $B(\gamma, r_\gamma)$, and
    \item the balls $B(\phi(\gamma), c\cdot r_{\phi(\gamma)})$ are pairwise disjoint.
\end{enumerate}

Consequently, due to property 2, we can upper bound $\sum_{\gamma \in \Gamma_1} \nval(B(\gamma, r_\gamma))$ by the sum of the values of the matched balls. According to property 1, this sum of values is at most the cost of the points in the ball in the solution $\Gamma$, as established in \Cref{lem:cover}. Additionally, property 3 ensures that there is no double counting, making this sum at most the cost of the entire dataset in the solution $\Gamma$.

In order to build this matching, the first step is to find $k$ balls that satisfy properties 2 and 3. To achieve this, we rely on the greedy choices made by the algorithm. 

For each $i = 1, \ldots, k$, let $(B(x_i^\ell, r_i^\ell))_\ell$ denote the sequence of balls $\sigma_i$, selected by the algorithm in the $i$-th loop.

Consider the balls $B(x_i^1, r_i^1)$ for $i = 1, \ldots, k$: each of these balls is chosen as the ball currently available with the maximum value. Therefore, they all have a value larger than that of $B(\gamma, r_\gamma)$, as this ball is still available at the end of the algorithm, thus satisfying property 3.

However, these balls may be too close to each other, and property 2 may not be satisfied. The idea is that if two balls $B(x_i^1, c \cdot r_i^1)$ and $B(x_{i'}^1, c \cdot r_{i'}^1)$ intersect (with $i' > i$), then we can preserve property 2 while reducing the diameter of one of the balls by considering $B(x_i^2, c \cdot r_i^2)$ instead. This approach is formalized and generalized in the next lemma:

\begin{lemma}\label{lem:pruningOneStep}
For every $i,i',\ell,\ell'$ such that $i<i'$ and $B(x_i^{\ell}, 2c\cdot r_i^{\ell}) \cap B(x_{i'}^{\ell'}, 2c\cdot r_{i'}^{\ell'}) \neq \emptyset$, it holds that:
\begin{itemize}
    \item $r_i^{\ell} \geq 4c^{2} \cdot   r_{i'}^{1} $
    \item $\nval(B(x_i^{\ell+1},r_i^{\ell+1})) \geq \nval(B(x_{i'}^1,r_{i'}^{1}))$.
\end{itemize}
\end{lemma}

\begin{proof}Let $i, i', \ell, \ell'$ be such that $i < i'$ and $B(x_i^{\ell}, 2c \cdot r_i^{\ell}) \cap B(x_{i'}^{\ell'}, 2c \cdot r_{i'}^{\ell'}) \neq \emptyset$. We start by proving the first point by contradiction: suppose that $r_i^{\ell} < 4c^{2} \cdot r_{i'}^{1}$. We then bound the distance between $x_{i'}^1$ and $c_i$ to show that $B(x_{i'}^1, r_{i'}^1)$ became unavailable when $c_i$ was selected, contradicting the fact that it was later selected by the algorithm.

Since $B(x_i^{\ell}, 2c \cdot r_i^{\ell}) \cap B(x_{i'}^{\ell'}, 2c \cdot r_{i'}^{\ell'}) \neq \emptyset$, we have 

\[\dist(x_i^{\ell}, x_{i'}^{\ell'}) \leq 2c \cdot r_i^{\ell} + 2c \cdot r_{i'}^{\ell'} \leq 8c^{3} \cdot r_{i'}^{1} + 2c \cdot r_{i'}^{1} = (8c^3 + 2c) \cdot r_{i'}^{1}.
\]

Moreover, applying \Cref{lem:distDesc} twice, we get 
\[
\dist(c_i, x_i^{\ell}) \leq 20c^2 \cdot r_i^{\ell} \leq 80c^4 \cdot r_{i'}^{1}
\]
and 
\[
\dist(x_{i'}^{\ell'}, x_{i'}^{1}) \leq 20c^2 \cdot r_{i'}^{1}.
\]

Combining these three inequalities using the triangle inequality, we obtain:

\begin{align*}
    \dist(c_i, x_{i'}^1) &\leq \dist(c_i, x_i^{\ell}) + \dist(x_i^{\ell}, x_{i'}^{\ell'}) + \dist(x_{i'}^{\ell'}, x_{i'}^{1}) \\
    &\leq 80c^4 \cdot r_{i'}^{1} + (8c^3 + 2c) \cdot r_{i'}^{1} + 20c^2 \cdot r_{i'}^{1} \\
    &\leq (80c^4 + 8c^3 + 20c^2 + 2c) \cdot r_{i'}^{1} \\
    &\leq 100c^4 \cdot r_{i'}^{1}.
\end{align*}

The last step follows from $c \geq 5$. Therefore, we have $x_{i'}^1 \in B(c_i, 100c^4 \cdot r_{i'}^1) \subseteq N(c_i, 100c^4 \cdot r_{i'}^1)$, and $B(x_{i'}^1, r_{i'}^1)$ is removed from the available balls after $c_i$ is selected, contradicting the fact that it was later picked by the algorithm.

We now turn to the second point. The inequality $r_i^{\ell} \geq 4c^{2} \cdot r_{i'}^{1}$ leads to $\dist(x_i^{\ell}, x_{i'}^{\ell'}) \leq 2c\cdot r_i^{\ell} + 2c\cdot r_{i'}^{\ell'} \leq 2c\cdot r_i^{\ell} + 2c\cdot r_{i'}^{1} \leq \left(2c + \frac{1}{2c^2}\right) r_i^{\ell}$. On the other hand, reusing the inequality given by \Cref{lem:distDesc}, we have $\dist(x_{i'}^{\ell'}, x_{i'}^{1}) \leq 20c^2 \cdot r_{i'}^{1} \leq 5\cdot  r_i^{\ell}$. Hence, using the triangle inequality, we get:

\begin{align*}
    \dist(x_i^{\ell}, x_{i'}^{1}) &\leq \dist(x_i^{\ell}, x_{i'}^{\ell'}) + \dist(x_{i'}^{\ell'}, x_{i'}^{1}) \\
    &\leq \left(2c + \frac{1}{2c^2}\right)\cdot  r_i^{\ell} + 5 \cdot r_i^{\ell} \\
    &\leq \left(2c + \frac{1}{2c^2} + 5\right)\cdot  r_i^{\ell}\\
    &\leq 10c \cdot r_i^\ell.
\end{align*}
The last step follows from $c \geq 5$. Therefore, we have $x_{i'}^{1} \in B(x_i^\ell, 10c \cdot r_i^\ell) \subseteq N(x_i^\ell, 10c \cdot r_i^\ell)$, and $B(x_{i'}^{1}, r_i^\ell / 2c)$ could have been selected by the algorithm instead of $B(x_i^{\ell+1}, r_i^{\ell+1})$. Hence,
\begin{align*}
    \nval(B(x_i^{\ell+1}, r_i^{\ell+1})) &\geq \nval(B(x_{i'}^{1}, r_i^\ell / 2c)) \\
    &\geq \nval(B(x_{i'}^1, r_{i'}^{1})).
\end{align*}

The last inequality comes from $r_i^{\ell} \geq 4c^{2} \cdot r_{i'}^{1}$ and \Cref{lem:nvaldecreasing}.
\end{proof}

\paragraph{Pruning the sequences:}

Let $M$ be the maximum value of balls that are still available at the end of the algorithm (if no ball is still available at the end of the algorithm, we can directly conclude with \Cref{cor:gamma0}). By definition, we have $\nval\lpar B(\gamma, r_\gamma)\rpar \leq M, \forall \gamma \in \Gamma_1$.

The next step to define the matching is to show that we can \emph{prune} all the $k$ sequences $(x_i^1, x_i^2, \dots)$ to establish a separation property. The pruning procedure, based on \Cref{lem:pruningOneStep}, removes some balls from each sequence, ensuring that the value of the first remaining ball in each sequence is at least $M$, while also guaranteeing that the remaining balls are sufficiently far apart from each other. This is formalized in the following lemma.

\begin{lemma}\label{lem:pruning}
    There exists indices $\ell_1, ..., \ell_k$ such that:
    \begin{itemize}
        \item for all $i \in \{1,k\}$, $\nval(B(x_i^{\ell_i},r_i^{\ell_i})) \geq M$.
        \item For all $i,i' \in  \{1,\dots,k\}$, and for all $\ell \geq \ell_i$, $\ell'\geq \ell_{i'}$, $B(x_i^{\ell}, 2c\cdot r_i^{\ell}) \cap B(x_{i'}^{\ell'}, 2c\cdot r_{i'}^{\ell'}) = \emptyset$.
    \end{itemize}
\end{lemma}
\begin{proof}
    
    Initially, set $\ell_i = 1$ for all $i$. This choice ensures that the first condition is satisfied: when $B(x_i^1,r_i^1)$ is selected, it maximizes the value among all available balls. Therefore $\nval(B(x_i^1,r_i^1)) \geq M$.

To satisfy the second condition, we follow this procedure: whenever there exist $i < i'$ and $\ell \geq \ell_i$, $\ell' \geq \ell_{i'}$ such that $B(x_i^{\ell}, 2c \cdot r_i^{\ell}) \cap B(x_{i'}^{\ell'}, 2c \cdot r_{i'}^{\ell'}) \neq \emptyset$, update $\ell_i$ to $\ell + 1$. According to the first item of \Cref{lem:pruningOneStep}, this procedure is well-defined because $B(x_i^{\ell}, r_i^{\ell})$ is not the last ball in the sequence, ensuring that $B(x_i^{\ell+1}, r_i^{\ell+1})$ exists. Additionally, the second item of the lemma guarantees that $\nval(B(x_i^{\ell+1}, r_i^{\ell+1})) \geq \nval(B(x_{i'}^1, r_{i'}^{\ell+1})) \geq M$, thereby maintaining the first condition after each update.

Since one of the $\ell_i$ is incremented at each step, the procedure must eventually terminate, as the maximum length of the sequences is $\log_{2c}(\Delta) + 7$. When the procedure ends, both conditions are satisfied, thus concluding the proof.
\end{proof}

\paragraph{Defining the matching:}
Starting from the pruned sequences and \Cref{lem:cover}, we can conclude the construction of the desired matching:

\begin{lemma}\label{lem:matching}
    There exists a matching $B(\gamma,r_\gamma) \mapsto B(\phi(\gamma),r_{\phi(\gamma)})$ defined for all $\gamma \in \Gamma_1$ such that:
     \begin{enumerate}
        \item $B(\phi(\gamma),r_{\phi(\gamma)})$ is not covered by $\Gamma$.
        \item For all $\gamma'$ with $\gamma \neq \gamma'$, $B(\phi(\gamma),c\cdot r_{\phi(\gamma)}) \cap B(\phi(\gamma'),c\cdot r_{\phi(\gamma')})) = \emptyset$.
        \item $\nval((B(\gamma,r_\gamma))\leq \nval(B(\phi(\gamma),r_{\phi(\gamma)}))$.
    \end{enumerate}
\end{lemma}
\begin{proof}
Let $\ell_i$ be the indices provided by \Cref{lem:pruning}. The construction of the matching proceeds in three steps:
\begin{itemize}
    \item First, note that if the last ball $B(c_i, 1/(2c)^7)$ of the $i$-th sequence is covered by an element $\gamma \in \Gamma$, then $\gamma \in \Gamma_0$ and we don't need to define the matching for $\gamma$.    
    \item Second, for any $i$ such that at least one ball in the pruned sequence $(B(x_i^{\ell}, r_i^{\ell}))_{\ell \geq \ell_i}$ is covered by $\Gamma$ but not the last one, we define $\lambda_i \geq \ell_i$ as the smallest index such that $B(x_i^{\ell}, r_i^{\ell})$ is not covered for all $\ell \geq \lambda_i$. Let $\gamma$ be an arbitrary element of $\Gamma$ that covers $B(x_i^{\lambda_i - 1}, r_i^{\lambda_i - 1})$. We then define $B(\phi(\gamma), r_{\phi(\gamma)}) = B(x_i^{\lambda_i}, r_i^{\lambda_i})$.
 
    \item Last, for any element $\gamma$ in $\Gamma_1$ that is still unmatched, we define $B(\phi(\gamma),r_{\phi(\gamma)}) = B(x_i^{\ell_i},r_i^{\ell_i})$, where $i$ is chosen arbitrarily such that none of the balls in the pruned sequence $(B(x_i^{\ell},r_i^{\ell}))_{\ell\geq \ell_i}$ are covered and such that the matching is one-to-one.
\end{itemize}

Note that the second item of Lemma \ref{lem:pruning} guarantees that if $\gamma\in \Gamma$ covers a ball of a pruned sequence, it cannot cover a ball of another pruned sequence: this ensures that our definition of the matching is consistent. We can now verify it satisfies the three desired properties.

\begin{enumerate}
\item For any $\gamma \in \Gamma_1$, $B(\phi(\gamma),r_{\phi(\gamma)})$ is not covered by $\Gamma$ by construction.
\item For any $\gamma, \gamma' \in \Gamma_1$, there exist indices $i, i'$ such that $i \neq i'$, $B(\phi(\gamma), r_{\phi(\gamma)})$ is a ball of the pruned sequence $(B(x_i^{\ell}, r_i^{\ell}))_{\ell \geq \ell_i}$, and $B(\phi(\gamma'), r_{\phi(\gamma')})$ is a ball of the pruned sequence $(B(x_{i'}^{\ell}, r_{i'}^{\ell}))_{\ell \geq \ell_{i'}}$. Therefore, the second item of \Cref{lem:pruning} ensures that $B(\phi(\gamma), 2c \cdot r_{\phi(\gamma)}) \cap B(\phi(\gamma'), 2c \cdot r_{\phi(\gamma')}) = \emptyset$.

\item Let $\gamma \in \Gamma_1$. We distinguish two cases, based on whether $B(\phi(\gamma), r_{\phi(\gamma)})$ was defined at the second or last step of the procedure.

If $ B(\phi(\gamma), r_{\phi(\gamma)}) $ is defined in the last step, then it is of the form $ B(x_i^{\ell_i}, r_i^{\ell_i}) $. By \Cref{lem:pruning}, we have $\nval(B(x_i^{\ell_i}, r_i^{\ell_i})) \geq M$. Combined with the fact that $ B(\gamma, r_\gamma) $ is available at the end of the algorithm, we directly obtain $\nval(B(\gamma, r_\gamma)) \leq \nval(B(x_i^{\ell_i}, r_i^{\ell_i}))$.

Otherwise, $ B(\gamma, r_\gamma) $ is defined in the second step, and $ B(\phi(\gamma), r_{\phi(\gamma)}) = B(x_i^{\lambda_i}, r_i^{\lambda_i}) $ for some $i$. We know that $\gamma$ covers $B(x_i^{\lambda_i-1}, r_i^{\lambda_i-1})$; therefore, $\gamma \in B(x_i^{\lambda_i-1}, 2c \cdot r_i^{\lambda_i-1}) \subseteq B(x_i^{\lambda_i-1}, 10c \cdot r_i^{\lambda_i-1}) \subseteq N(x_i^{\lambda_i-1}, 10c \cdot r_i^{\lambda_i-1})$. Hence, the algorithm could have picked $ B(\gamma, r_i^{\lambda_i}) $ instead of $ B(x_i^{\lambda_i}, r_i^{\lambda_i}) $, and therefore $\nval(B(\gamma, r_i^{\lambda_i})) \leq \nval(B(x_i^{\lambda_i}, r_i^{\lambda_i}))$.

It remains to prove that $ r_\gamma \leq r_i^{\lambda_i} $ to conclude with \Cref{lem:nvaldecreasing}.

Assume, for contradiction, that $ r_i^{\lambda_i} < r_\gamma $. Because $\gamma$ covers $ B(x_i^{\lambda_i-1}, r_i^{\lambda_i-1}) $, we know that $\dist(\gamma, x_i^{\lambda_i-1}) \leq 2c \cdot r_i^{\lambda_i-1} = 4c^2 \cdot r_i^{\lambda_i} < 4c^2 \cdot r_\gamma $. Moreover, by \Cref{lem:distDesc}, we have $\dist(x_i^{\lambda_i-1}, c_i) \leq 20c^2 \cdot r_i^{\lambda_i-1} = 40c^3 \cdot r_i^{\lambda_i} < 40c^3 \cdot r_\gamma$. Hence, using the triangle inequality, we get:
\begin{align*}
    \dist(\gamma, c_i) & \leq \dist(\gamma, x_i^{\lambda_i-1}) + \dist(x_i^{\lambda_i-1}, c_i) \\
    & < (4c^2 + 40c^3) \cdot r_\gamma \\
    & < 100c^4 \cdot r_\gamma.
\end{align*}
Therefore, $\gamma \in B(c_i, 100c^4 r_\gamma) \subseteq N(c_i, 100c^4 r_\gamma)$. Thus, $ B(\gamma, r_\gamma) $ is removed from the set of available balls after $ c_i $ is selected, contradicting the fact that $ B(\gamma, r_\gamma) $ is still selected at the end of the algorithm. This completes the proof that $ r_\gamma \leq r_i^{\lambda_i} $.

Now, applying \Cref{lem:nvaldecreasing}, we get $\nval(B(\gamma, r_\gamma)) \leq \nval(B(\gamma, r_i^{\lambda_i}))$. Combining this with the inequality $\nval(B(\gamma, r_i^{\lambda_i})) \leq \nval(B(x_i^{\lambda_i}, r_i^{\lambda_i}))$ obtained earlier, we conclude the proof.

\end{enumerate}
\end{proof}

\subsection[Putting things together: proof of Theorem 5.1]{Putting things together: proof of \Cref{thm:correctness}}

We conclude the proof of our main theorem as follows: 
\begin{proof}[Proof of \Cref{thm:correctness}]

Given the matching $\phi$ of \Cref{lem:matching}, we can conclude as follows. Summing the inequality of the third property of $\phi$ gives

\[
\sum_{\gamma \in \Gamma_1} \nval(B(\gamma, r_\gamma)) \leq \sum_{\gamma \in \Gamma_1} \nval(B(\phi(\gamma), r_{\phi(\gamma)})).
\]

Each $B(\phi(\gamma), r_{\phi(\gamma)})$ is not covered by $\Gamma$ by the first property of $\phi$. Therefore, we can apply \Cref{lem:cover} and obtain

\[
\sum_{\gamma \in \Gamma_1} \nval(B(\phi(\gamma), r_{\phi(\gamma)})) \leq \frac{3}{c^z} \cdot \sum_{\gamma \in \Gamma_1} \cost(B(\phi(\gamma), c\cdot r_{\phi(\gamma)}) \cap P, \Gamma).
\]

The second property of $\phi$ ensures that the balls in the set $\{ B(\phi(\gamma), c\cdot r_{\phi(\gamma)}) \mid \gamma \in \Gamma_1 \}$ are disjoint. Therefore,

\[
\sum_{\gamma \in \Gamma_1} \cost(B(\phi(\gamma), c\cdot r_{\phi(\gamma)}) \cap P, \Gamma) = \cost\left(\bigcup_{\gamma \in \Gamma_1} B(\phi(\gamma), c\cdot r_{\phi(\gamma)}) \cap P, \Gamma\right).
\]

Combining everything yields

\[
\sum_{\gamma \in \Gamma_1} \nval(B(\gamma, r_\gamma)) \leq \frac{3}{c^z} \cdot \cost\left(\bigcup_{\gamma \in \Gamma_1} B(\phi(\gamma), c\cdot r_{\phi(\gamma)}) \cap P, \Gamma\right).
\]
Combining this inequality with \Cref{lem:inAndout} and \Cref{lem:gamma0} finishes the proof of \Cref{thm:correctness}:
\begin{align*}
    \cost(P,C_k) &= \sum_{\gamma \in \Gamma_0}\cost(P_\gamma),C_k) + \sum_{\gamma \in \Gamma_1}\cost(Out(P_\gamma),C_k) + \sum_{\gamma \in \Gamma_1}\cost(Out(P_\gamma),C_k)\\
    &\leq \sum_{\gamma \in \Gamma_0} \cost(P_\gamma, \Gamma) + 2^{z-1}((200c^6)^z+1)\cdot \sum_{\gamma \in \Gamma_1} \cost(Out(P_\gamma),\Gamma) \\
    &\qquad +2^{z-1}\lpar(200c^6)^z + 1 \rpar \cdot 3   \cdot \sum_{\gamma \in \Gamma_1} \nval(B(\gamma,r_\gamma))\\
    &\leq \sum_{\gamma \in \Gamma_0} \cost(P_\gamma, \Gamma) + 2^{z-1}((200c^6)^z+1)\cdot \sum_{\gamma \in \Gamma_1} \cost(Out(P_\gamma),\Gamma)\\
    &\qquad +2^{z-1}\lpar(200c^6)^z + 1 \rpar \cdot 3   \cdot \frac{3}{c^z}\cdot \sum_{\gamma \in \Gamma_1} \cost(\bigcup_{\gamma \in \Gamma_1} B(\phi(\gamma), c\cdot r_{\phi(\gamma)}) \cap P, \Gamma)\\
    &\leq (1+ 2^{z-1}((200c^6)^z+1)(1+\frac{9}{c^z}))\cost(P,\Gamma).
\end{align*}

\end{proof}

\section{Reducing the diameter} 
\subsection{In Euclidean space} 
For $(k,z)$-clustering, we can assume that the aspect ratio $\Delta = \poly(n)$: \cite{reduceDiam} showed how to transform any input $P$ to reduce the diameter. Their algorithm runs in time $O\lpar nd \log \log \Delta\rpar$, which is the running-time of their algorithm to compute a $\poly(n)$-approximation. The later has been improved to $\tilde O(nd)$ by \cite{charikar2023simple}, hence we can reduce to $\Delta = \poly(n)$ in time $\tilde O(nd)$.

\subsection{In sparse graphs}
We describe a simple a preprocessing step that reduces the diameter $\Delta$ to $\poly(n)$. The algorithm can be described as follows: compute a minimum spanning tree, and consider the weight  $w_k$ of the $k$-th largest edge in it. As we will show, the optimal  $(k, z)$-clustering has cost at least $w_k^z/ n^z$ and at most $n^{z+1} w_k^z$. Because every constant‑factor approximation to the optimal cost stays not too far from this interval, it can never use an edge whose weight is way larger than $n w_k$. We therefore cap such edges at a large value that still keeps them useless for any close-to-optimal solution, while reducing the overall diameter. On the other hand, accuracy at a degree much finer than $w_k$ is irrelevant, so we round every remaining edge weight up to the nearest multiple of $w_k/n^3$. Finally, we rescale the metric so that all edge weights are integers in the range $\{1,\dots,n^6\}.$

\begin{algorithm}[H]
\caption{Reducing $\Delta$}
\label{alg:poly(n)}
\textbf{Input:} A weighted connected graph $G = (V,E,w)$ with $|E| = m$ and $|V| = n$, and a number of clusters $k \leq n-1$.\\
\textbf{Output:} A new weight function $w^*$ on $E$ with polynomial diameter.
\begin{algorithmic}[1]
\State{Compute a Minimum Spanning Tree $T$ of $G$.}
\State{Let $w_k$ be the weight of the $k$-th largest edge in $T$.}
\If{$w_k \leq n^3$}
\For{$e\in E$}
\State{Set $w^*(e) := \min(w(e), n^{6})$.}
\EndFor
\Else
\For{$e\in E$}
\State{Set $w^*(e) := \min(\lceil w(e) \cdot \frac{n^{3}}{w_k}\rceil, n^{6})$.}
\EndFor
\EndIf
\State{Output $w^*$.}
\end{algorithmic}
\end{algorithm}

\begin{lemma}
Algorithm~\ref{alg:poly(n)} can be implemented in time $\tilde{O}(m)$. The diameter $\Delta^*$ induced by $w^*$ is less than $n^7$. Moreover, for any constant $\alpha$ and for $n$ large enough, if a set $C$ of $k$ centers is an $\alpha$-approximation for $(k,z)$-clustering in $G^* = (V,E,w^*)$, then it is an $(\alpha \cdot 2^z)$- approximation for $(k,z)$-clustering in $G = (V,E,w)$. 
\end{lemma}
\begin{proof}\noindent
Using any efficient MST algorithm (e.g., Kruskal's~\cite{Kruskal}), we can compute the minimum spanning tree $T$ of $G$ in $\tilde{O}(m)$ time. To bound the new diameter $\Delta^*$, note that for each edge $e \in E$, we have $w^*(e) \leq n^6$, implying $\Delta^* \leq n^7$. We now turn to bounding the approximation ratio.

Let $\opt_G$ and $\opt_T$ denote the optimal $(k,z)$-clustering costs in $G$ and $T$, respectively. Since $T$ is a subgraph of $G$, it follows immediately that $\opt_G \;\le\; \opt_T$.

Because $T$ is a minimum spanning tree, for any edge $(u,v) \in E$ there is a path in $T$ whose edges all have weights at most $w(u,v)$. Thus, the distance between $u$ and $v$ in $T$ is at most $n \cdot w(u,v)$, and by induction on the number of edges in a path, for any $u,v\in V$ we have:
\[
  \dist_T(u,v) \;\le\; n \cdot \dist_G(u,v).
\]
It follows that $\opt_T \;\le\; n^z \cdot \opt_G$.

\noindent

Removing the $k-1$ largest edges from $T$ results in $k$ connected components. By choosing one center in each component, any vertex $u$ in $T$ is connected to its center by a path free of the removed edges, so the maximum distance between a vertex and its center is at most $n\cdot w_k$, where $w_k$ is the weight of the $k$-th largest edge in $T$. Hence, each vertex contributes at most $(n\,w_k)^z$ to the cost, and 
\[
  \textstyle
  \sum_{u\in V} \dist(u, \mathrm{center})^z \;\le\; n \cdot (n\,w_k)^z  \;\le\; n^{z+1} \cdot w_k^z.
\]
Since $\opt_G \leq \opt_T$, we conclude $\opt_G \;\le\; n^{z+1} \cdot w_k^z$.

On the other hand, forming $k$ clusters in $T$ requires using at least one of the top-$k$ heaviest edges, i.e., the solution's paths from vertices to centers must include at least one such edge. Otherwise, we would have $k+1$ connected components, contradicting the fact that there are only $k$ clusters. Consequently,
\[
  \opt_T \;\ge\; w_k^z
  \quad\Longrightarrow\quad
  \opt_G \;\ge\; \frac{w_k^z}{n^z}.
\]

\noindent

Assume we have a set of $k$ centers $C$ that provides an $\alpha$-approximation for the optimal $(k,z)$-clustering cost in $G^*$. We analyze two cases based on the value of $w_k$.

\noindent
\textbf{Case 1: $w_k \leq n^3$.} 

In this case, for any edge $e$, $w^*(e)  \le w(e)$ and thus $\opt_{G^*} \;\le\; \opt_G$. By the $\alpha$-approximation in $G^*$, 
\[
  \cost_{G^*}(C) \;\le\; \alpha \cdot \opt_{G^*} \;\le\; \alpha \cdot \opt_G.
\]
Suppose (toward a contradiction) that some edge $e$ on the path from a vertex $u$ to its center in $C$ has $w^*(e) = n^{6}$. Then 
\[
  n^{6 \cdot z} \;\le\; \cost_{G^*}(C) \;\le\; \alpha\cdot \opt_G \;\le\; \alpha \cdot n^{z+1}\, w_k^z
  \;\le\; \alpha \cdot n^{4z+1},
\]
which implies $n \le \alpha$. For sufficiently large $n$, this is a contradiction. Therefore, for all edges used in $\cost_{G^*}(C)$, we must have $w^*(e) = w(e)$. Hence 
\[
  \cost_G(C) \;=\; \cost_{G^*}(C) 
  \;\le\; \alpha \cdot \opt_{G},
\]
so $C$ is also an $\alpha$-approximation in $G$.

\noindent
\textbf{Case 2: $w_k > n^3$.} 
Now $w^*(e)\le w(e) \cdot \frac{n^{3}}{w_k} +1$ for any edge $e$. Let $C_{\opt}$ be an optimal $(k,z)$-clustering in $G$. Any vertex $u$ is at most $n$ edges away from its center in $C_{\opt}$, so
\[
  \dist_{w^*}(u, C_{\opt}) 
  \;\le\; \dist_w(u, C_{\opt}) \cdot \frac{n^{3}}{w_k} \;+\; n.
\]
Hence,

\begin{align*}
  \opt_{G^*} 
  &\;\le\; \sum_{u\in V} \dist_{w^*}(u, C_{\opt})^z\\
  &\;\le\; \sum_{u \in V} 
    \Bigl(\dist_w(u, C_{\opt})\cdot \frac{n^{3}}{w_k} + n\Bigr)^z\\
  &\;\le\;
    \sum_{u \in V} 2^{\,z-1}\!\Bigl[
      \Bigl(\dist_w(u, C_{\opt}) \cdot \frac{n^{3}}{w_k}\Bigr)^z + n^z
    \Bigr]\\
  &\;\le\; 
    2^{\,z-1} \Bigl(\opt_G \cdot (\frac{n^{3}}{w_k})^z + n^{\,z+1}\Bigr).
\end{align*}
Using the $\alpha$-approximation in $G^*$,
\[
  \cost_{G^*}(C) \;\le\; \alpha \cdot \opt_{G^*}
  \;\le\; \alpha \cdot 2^{\,z-1}\!
    \Bigl(\opt_G \cdot (\frac{n^{3}}{w_k})^z + n^{z+1}\Bigr).
\]
Suppose there is a vertex $u$ whose path to its center in $C$ uses an edge $e$ with $w^*(e) = n^{6}$. Then
\[
  n^{6 z} 
  \;\le\; \cost_{G^*}(C)
  \;\le\; \alpha \cdot 2^{\,z-1}\!
    \Bigl(\opt_G \cdot (\frac{n^{3}}{w_k})^z + n^{z+1}\Bigr).
\]
Since $\opt_G \leq n^{z+1}\, w_k^z$, we get
\[
  n^{6z} 
  \;\le\; \alpha \cdot 2^{\,z-1}\!
    \Bigl(n^{z+1}\, w_k^z \cdot (\frac{n^{3}}{w_k})^z + n^{z+1}\Bigr)
  \;\Longrightarrow\;
  n^{5z -1} 
  \;\le\; \alpha \cdot 2^{\,z-1}\!
    \Bigl(n^{3z} + 1\Bigr).
\]

which is a contradiction for large $n$. Therefore, $w^*(e) = \lceil w(e) \cdot \frac{n^{3}}{w_k}\rceil \ge  w(e) \cdot \frac{n^{3}}{w_k}$ for all edges used in $\cost_{G^*}(C)$, and so
\[
  \cost_{G}(C) 
  \;\le\; (\frac{w_k}{n^{3}})^z \cdot \cost_{G^*}(C).
\]
Combining this with our earlier bound,
\[
  \cost_{G}(C) 
  \;\le\; (\frac{w_k}{n^{3}})^z \cdot 
    \Bigl[\alpha \cdot 2^{\,z-1}\!
    \Bigl(\opt_G \cdot (\frac{n^{3}}{w_k})^z + n^{z+1}\Bigr)\Bigr]
  \;=\; \alpha \cdot 2^{\,z-1}
    \Bigl(\opt_G + (\frac{w_k}{n^{3}})^z\, n^{z+1}\Bigr).
\]
We have
\[
  (\frac{w_k}{n^{3}})^z\, n^{z+1} 
  \;\le\;  
    \Bigl(\tfrac{w_k}{n}\Bigr)^z
  \;\le\; \opt_G,
\]
Putting it all together,
\[
  \cost_{G}(C) 
  \;\le\;
  \alpha \cdot 2^{\,z-1} \bigl(\opt_G + \opt_G\, \bigr) 
   \;\le\;  \alpha \cdot 2^z\, \opt_G.
\]
Hence, $C$ is an $\alpha \cdot 2^z$-approximation of the optimal $(k,z)$-clustering cost in $G$.

\end{proof}

\end{document}